\documentclass[a4paper,fleqn]{llncs}
\usepackage[utf8]{inputenc}

\usepackage{algorithm}
\usepackage{algorithmicx}
\usepackage[noend]{algpseudocode}
\usepackage{amsmath}
\usepackage{amsfonts}
\usepackage{amssymb}
\usepackage{booktabs}
\usepackage{xspace,enumitem}
\usepackage{marvosym}
\usepackage{comment}
\usepackage{thm-restate}
\usepackage{doi}

\usepackage{tikz}
\usetikzlibrary{automata,arrows,decorations.pathreplacing,plotmarks,shapes.geometric}
\usepackage{pgfplots}
\usepackage{pgfplotstable}
\pgfplotsset{compat=1.17}

\newcommand{\eg}{\emph{e.g.}\xspace}

\newcommand{\ie}{\emph{i.e.}\xspace}

\newcommand{\nats}{\mathbb{N}}
\newcommand{\podd}{\ensuremath{\square}\xspace}
\newcommand{\peven}{\ensuremath{\Diamond}\xspace}
\newcommand{\player}{\ensuremath{\alpha}\xspace}
\newcommand{\edgerel}{\ensuremath{E}}
\newcommand{\otherplayer}{\ensuremath{\bar\alpha}\xspace}
\newcommand{\priof}{\ensuremath{p}\xspace}

\newcommand{\pre}[3][]{\ensuremath{\textsf{pre}(#2,#3)}}
\newcommand{\cpre}[4][]{\ensuremath{\textsf{cpre}_{#2}(#3,#4)}}
\newcommand{\spre}[4][]{\ensuremath{\textsf{spre}_{#2}(#3,#4)}}
\newcommand{\sinks}[1]{\textsf{sinks}(#1)}
\newcommand{\psinks}[2]{\textsf{sinks}_{#1}(#2)}
\newcommand{\safe}[2]{\textsf{safe}_{#1}(#2)}

\newcommand{\attr}[4][]{\ensuremath{\textsf{Attr}_{#2}(#3,#4)}}

\newcommand{\sattr}[4][]{\ensuremath{\textsf{SAttr}_{#2}(#3,#4)}}

\usepackage{todonotes}

\newcommand{\solcycle}[2]{\mathcal{C}^{#1}_{\mathsf{sol}}(#2)}
\newcommand{\allcycle}[2]{\mathcal{C}^{#1}_{\mathsf{for}}(#2)}
\newcommand{\sallcycle}[2]{\mathcal{C}^{#1}_{\mathsf{s-for}}(#2)}
\newcommand{\Mattr}[4]{\ensuremath{\textsf{MAttr}_{#1}(#2,#3,#4)}}
\newcommand{\sMattr}[4]{\ensuremath{\textsf{sMAttr}_{#1}(#2,#3,#4)}}
\newcommand{\Mpre}[5]{\ensuremath{\textsf{Mcpre}_{#1}(#2,#3,#4,#5)}}
\newcommand{\sMpre}[5]{\ensuremath{\textsf{sMcpre}_{#1}(#2,#3,#4,#5)}}
\newcommand{\fattr}[3]{\mathcal{F}^{#1}(#2,#3)}
\newcommand{\sfattr}[3]{\mathcal{F}^{#1}_{\mathsf{s}}(#2,#3)}

\newcommand{\timeout}{\ensuremath{\ddagger}}

\title{On-The-Fly Solving for Symbolic Parity Games}
\author{Maurice~Laveaux\inst{1} (\Letter)
\and Wieger~Wesselink\inst{1}
\and Tim~A.C.~Willemse\inst{1,2}}

\institute{Eindhoven University of Technology, The Netherlands \and
ESI (TNO), Eindhoven, The Netherlands \\
\email{\{m.laveaux, j.w.wesselink, t.a.c.willemse\}@tue.nl}}

\begin{document}

\maketitle

\begin{abstract}
  Parity games can be used to represent many different kinds of decision problems.
  In practice, tools that use parity games often rely on a specification in a higher-order logic from which the actual game can be obtained by means of an exploration.
  For many of these decision problems we are only interested in the solution for a designated vertex in the game.
  We formalise how to use on-the-fly solving techniques during the exploration process, and show that this can help to decide the winner of such a designated vertex in an incomplete game.
  Furthermore, we define partial solving techniques for incomplete parity games and show how these can be made resilient to work directly on the incomplete game, rather than on a set of \emph{safe} vertices.
  We implement our techniques for symbolic parity games and study their effectiveness in practice, showing that speed-ups of several orders of magnitude are feasible and overhead (if unavoidable) is typically low.
\end{abstract}

\section{Introduction}
A parity game is a two-player game with an $\omega$-regular winning condition, played by players $\peven$ (`even') and $\podd$ (`odd') on a directed graph. 
The true complexity of solving parity games is still a major open problem, with the most recent breakthroughs yielding algorithms running in quasi-polynomial time, see,~\eg,~\cite{JurdzinskiL17,CaludeJKL017}.
Apart from their intriguing status, parity games pop up in various fundamental results in computer science (\eg, in the proof of decidability of a monadic second-order theory).
In practice, parity games provide an elegant, uniform framework to encode many relevant decision problems, which include model checking problems, synthesis problems and behavioural equivalence checking problems.

Often, a decision problem that is encoded as a parity game, can be answered by determining which of the two players wins a designated vertex in the game graph.
Depending on the characteristics of the game, it may be the case that only a fraction of the game is relevant for deciding which player wins a vertex.
For instance, deciding whether a transition system satisfies an invariant can be encoded by a simple, solitaire (\ie, single player) parity game.
In such a game, player $\podd$ wins all vertices that are sinks (\ie, have no successors), and all states leading to such sinks, so checking whether sinks are reachable from a designated vertex suffices to determine whether this vertex is won by $\podd$, too.
Clearly, as soon as a sink is detected, any further inspection of the game becomes irrelevant.

A complicating factor is that in practice, the parity games that encode decision problems are not given explicitly.
Rather, they are specified in some higher-order logic such as a parameterised Boolean equation system, see, \eg~\cite{CranenLW13}.
Exploring the parity game from such a higher-order specification is, in general, time-and memory-consuming.
To counter this, symbolic exploration techniques have been proposed, see~\eg~\cite{KantP12}.
These explore the game graph on-the-fly and exploit efficient symbolic data structures such as LDDs~\cite{DijkP17} to represent sets of vertices and edges.
Many parity game solving algorithms can be implemented quite effectively using such data structures~\cite{KantP14,SanchezWW18,StasioMV18}, so that in the end, exploring the game graph often remains the bottleneck.

In this paper, we study how to combine the exploration of a parity game and the on-the-fly solving of the explored part, with the aim to speed-up the overall solving process.
The central problem when performing on-the-fly solving during the exploration phase is that we have to deal with incomplete information when determining the winner for a designated vertex.
Moreover, in the symbolic setting, the exploration order may be unpredictable when advanced strategies such as \emph{chaining} and \emph{saturation}~\cite{CiardoMS06} are used.

To formally reason about all possible exploration strategies and the artefacts they generate, we introduce the concept of an \emph{incomplete parity game}, and an ordering on these.
Incomplete parity games are parity games where for some vertices not all outgoing edges are necessarily known.
In practice, these could be identified by, \eg, the \emph{todo} queue in a classical breadth-first search.
The extra information captured by an incomplete parity game allows us to characterise the \emph{safe} set for a given player $\player$.
This is a set of vertices for which it can be established that if player $\player$ wins the vertex, then she cannot lose the vertex if more information becomes available.
We prove an optimality result for safe sets, which, informally, states that a safe set for player $\player$ is also the largest set with this property (see Theorem~\ref{th:optimality}).

The vertices won by player $\player$ in an $\player$-safe set can be determined using a standard parity game solving algorithm such as, \eg, 
Zielonka's recursive algorithm~\cite{Zielonka98} or 
Priority Promotion~\cite{BenerecettiDM18}.
However,  these algorithms may be less efficient as on-the-fly solvers.
For this reason, we study three symbolic \emph{partial} solvers: \emph{solitaire winning cycle} detection, \emph{forced winning cycle} detection and \emph{fatal attractors}~\cite{HuthKP13}.
In particular cases, first determining the safe set for a player and only subsequently solving the game using one of these partial solvers will incur an additional overhead.
As a final result, we therefore prove that all these solvers can be (modified to) run on the incomplete game as a whole, rather than on the safe set of a player (see Propositions~\ref{prop:solcycle}-\ref{prop:fatal}).

As a proof of concept, we have implemented an (open source) symbolic tool for the mCRL2 toolset~\cite{BunteGKLNVWWW19}, that explores a parity game specified by a parameterised Boolean equation system and solves these games on-the-fly.
We report on the effectiveness of our implementation on typical parity games stemming from, \eg, model checking and equivalence checking problems, showing that it can speed up the process with several orders of magnitude, while adding low overhead if the entire game is needed for solving.

\paragraph{Related Work.}
Our work is related to existing techniques for solving symbolic parity games such as~\cite{KantP14,KantP12}, as we extend these existing methods with on-the-fly solving.
Naturally, our work is also related to existing work for on-the-fly model checking.
This includes work for on-the-fly (explicit) model checking of regular alternation-free modal mu-calculus formulas~\cite{MateescuS03} and work for on-the-fly symbolic model checking of RCTL~\cite{BeerBL98}.
Compared to these our method is more general as it can be applied to the full modal mu-calculus (with data), which subsumes RCTL and the alternation-free subset.
Optimisations such as the observation that checking LTL formulas of type \textsf{AG} reduces to reachability checks~\cite{EirikssonM95} are a special case of our methods and partial solvers.
Furthermore, our methods are not restricted to model checking problems only and can be applied to any parity game, including decision problems such as equivalence checking~\cite{ChenPPW07}.
Furthermore, our method is agnostic to the exploration strategy employed.

\paragraph{Structure of the paper.}
In Section~\ref{section:preliminaries} we recall parity games.
In Section~\ref{section:incomplete_games} we introduce incomplete parity games and show how partial solving can be applied correctly.
In Section~\ref{sec:partial_solvers} we present several partial solvers that we employ for on-the-fly solving.
Finally, in Section~\ref{section:experiments} we discuss the implementation of these techniques and apply them to several practical examples.

\section{Preliminaries}\label{section:preliminaries}

A parity game is an infinite-duration, two-player game that is played on a finite directed graph. 
The objective of the two players, called \emph{even} (denoted by \peven) and \emph{odd} (denoted by \podd), is to win vertices in the graph.
\begin{definition}
	\label{def:parity game}
	A \emph{parity game} is a directed graph $G = (V,\edgerel,\priof,(V_\peven,V_\podd))$, where
	\begin{itemize}
		\item $V$ is a finite set of vertices, partitioned in sets $V_\peven$ and $V_\podd$ of vertices owned by \peven and \podd, respectively;
		\item $\edgerel \subseteq V \times V$ is the edge relation;
		\item $\priof: V \to \nats$ is a function that assigns a \emph{priority} to each node.
	\end{itemize}
\end{definition}
Henceforth, let $G = (V, \edgerel, \priof, (V_\peven,V_\podd))$ be an arbitrary parity game.
Throughout this paper, we use \player to denote an arbitrary player and \otherplayer denotes the opponent.
We write $v E$ to denote the set of successors $\{w \in V \mid (v,w) \in E \}$ of vertex $v$.
The set $\sinks{G}$ is defined as the largest set $U \subseteq V$ satisfying for all $v \in U$ that $v E = \emptyset$; \ie, $\sinks{G}$ is the set of all sinks: vertices without successors.
If we are only concerned with the sinks of player $\player$, we write $\psinks{\player}{G}$; \ie, $\psinks{\player}{G} = V_\player \cap \sinks{G}$.
We write $G \cap U$, for $U \subseteq V$, to denote the subgame $(U, (U \times U) \cap E, \priof\!\!\restriction_U, (V_\peven \cap U, V_\podd \cap U))$, where $\priof\!\!\restriction_U(v) = \priof(v)$ for all vertices $v \in U$.

\begin{figure}[h]
\centering
\scriptsize
	\begin{tikzpicture}[->,>=stealth',shorten >=0pt,auto,node distance=2.0cm,semithick]
	\tikzstyle{even}=[draw,inner sep=2pt,diamond]
	\tikzstyle{odd}=[draw,inner sep=2pt,regular polygon,regular polygon sides=4]		
	\node[even,label=below:$u_0$] (s1)  {2};
	\node[odd,label=below:$u_1$,left of=s1] (s0)  {3};
	\node[even,label=below:$u_2$,right of=s1] (s2) {0};
	\node[odd,label=below:$u_3$,right of=s2]  (s3)  {1};
	\node[even,label=below:$u_4$,right of=s3] (s4) {2};
	\path
        (s1) edge (s0)
	(s1) edge[bend left=28] (s2)
	(s2) edge[bend left=28] (s1)
	(s2) edge[bend left=38] (s4)
	(s3) edge               (s2)
	(s3) edge[bend left=28] (s4)
	(s4) edge[bend left=28] (s3)
	;
	\end{tikzpicture}
  \caption{An example parity game}\label{fig:parity_game}
\end{figure}
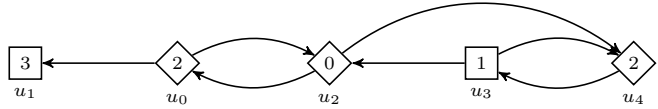
\begin{example}
	Consider the graph depicted in Figure~\ref{fig:parity_game}, representing a parity game.
        Diamond-shaped vertices are owned by player~\peven, whereas box-shaped vertices are owned by player~\podd.
	The priority of a vertex is written inside the vertex.
        Vertex $u_1$ is a sink owned by player $\podd$.\qed
\end{example}

\paragraph{Plays and strategies.}
The game is played as follows. 
Initially, a token is placed on a vertex of the graph.
The owner of a vertex on which the token resides gets to decide the successor vertex (if any) that the token is moved to next.
A maximal sequence of vertices (\ie, an infinite sequence or a finite sequence ending in a sink) visited by the token by following this simple rule is called a \emph{play}.
A finite play $\pi$ is won by player \peven if the sink in which it ends is owned by player \podd, and it is won by player \podd if the sink is owned by player \peven.
An infinite play $\pi$ is won by player \peven if the minimal priority that occurs infinitely often along $\pi$ is even, and it is won by player \podd otherwise.

A strategy $\sigma_\player: V^*V_\player \to V$ for player \player is a partial function that prescribes where player \player moves the token next, given a sequence of vertices visited by the token.
A play $v_0\,v_1\,\dots$ is \emph{consistent} with a strategy $\sigma$ if and only if $\sigma(v_0 \dots v_i) = v_{i+1}$ for all $i$ for which $\sigma(v_0 \dots v_i)$ is defined.
Strategy $\sigma_\player$ is winning for player $\player$ in vertex $v$ if all plays consistent with $\sigma_\player$ and starting in $v$ are won by $\player$.
Player \player wins vertex $v$ if and only if she has a winning strategy $\sigma_\player$ for vertex $v$.
The \emph{parity game solving problem} asks to compute the set of vertices $W_\peven$, won by player \peven and the set $W_\podd$, won by player $\podd$.
Note that since parity games are \emph{determined}~\cite{Zielonka98,McNaughton93}, every vertex is won by one of the two players.
That is, the sets $W_\peven$ and $W_\podd$ partition the set $V$.
\begin{example}
        Consider the parity game depicted in Figure~\ref{fig:parity_game}.
	In this game, the strategy $\sigma_\peven$, partially defined as $\sigma_\peven(\pi u_0) = u_2$ and $\sigma_\peven(\pi u_2) = u_0$, for arbitrary $\pi$, is winning for player~\peven in $u_0$ and $u_2$.
	Player~\podd wins vertex $u_3$ using strategy $\sigma_\podd(\pi u_3) = u_4$, for arbitrary $\pi$.
	Note that player~\peven is always forced to move the token from $u_4$ to $u_3$.
        Vertex $u_1$ is a sink, owned by player $\podd$, and hence, won by player $\peven$. \qed
\end{example}

\paragraph{Dominions.}

A strategy $\sigma_\player$ is said to be \emph{closed} on a set of vertices $U \subseteq V$ iff every play, consistent with $\sigma_\player$ and starting in a vertex $v \in U$ remains in $U$.
If player $\player$ has a strategy that is closed on $U$, we say that the set $U$ is $\player$-closed.
A \emph{dominion} for player $\player$ is a set of vertices $U \subseteq V$ such that player $\player$ has a strategy $\sigma_\player$ that is closed on $U$ and which is winning for $\player$.
Note that the sets $W_{\peven}$ and $W_{\podd}$ are dominions for player $\peven$ and player~$\podd$, respectively, and, hence, every vertex won by player $\player$ must belong to an $\player$-dominion.
\begin{example}
Reconsider the parity game of Figure~\ref{fig:parity_game}.
Observe that player~\podd has a closed strategy on $\{u_3,u_4\}$, which is also winning for player~$\podd$.
Hence, the set $\{u_3, u_4\}$ is an \podd-dominion.
Furthermore, the set $\{u_2, u_3, u_4\}$ is \peven-closed.
However, none of the strategies for which $\{u_2,u_3,u_4\}$ is closed for player $\peven$ is winning for her; therefore $\{u_2, u_3, u_4\}$ is not an $\peven$-dominion.\qed
\end{example}

\paragraph{Predecessors, control predecessors and attractors.}
Let $U \subseteq V$ be a set of vertices. We write $\pre{G}{U}$ to denote the set of predecessors $\{v \in V \mid \exists u \in U: u \in vE \}$ of $U$ in $G$.
The control predecessor set of $U$ for player $\player$ in $G$, denoted
$\cpre{\player}{G}{U}$, contains those vertices for which $\player$ is able to \emph{force} entering $U$ in one step.
It is defined as follows:
\begin{equation*}
\cpre{\player}{G}{U} = (V_{\player} \cap \pre{G}{U}) \cup (V_{\otherplayer} \setminus (\pre{G}{V \setminus U} \cup \sinks{G}))
\end{equation*}
Note that both $\textsf{pre}$ and $\textsf{cpre}$ are monotone operators on the complete lattice $(2^V, \subseteq)$.
The $\player$-attractor to $U$ in $G$, denoted $\attr{\player}{G}{U}$, is the set of vertices from which player $\player$ can force play to reach a vertex in $U$:
\begin{equation*}
\attr{\player}{G}{U} = \mu Z.(U \cup \cpre{\player}{G}{Z})
\end{equation*}
The $\player$-attractor to $U$ can be computed by means of a fixed point iteration, starting at $U$ and adding $\player$-control predecessors in each iteration until a stable set is reached.
We note that the $\player$-attractor to an $\player$-dominion $D$ is again an $\player$-dominion.
\begin{example}
Consider the parity game $G$ of Figure~\ref{fig:parity_game} once again.
The $\peven$-control predecessors of $\{u_2\}$ is the set $\{u_0\}$.
Note that since player $\podd$ can avoid moving to $u_2$ from vertex $u_3$ by moving to vertex $u_4$, vertex $u_3$ is not among the $\peven$-control predecessors of $\{u_2\}$.
The $\peven$-attractor to $\{u_2\}$ is the set $\{u_0,u_2\}$, which is the largest set of vertices for which player $\peven$ has a strategy to force play to the set of vertices $\{u_2\}$.\qed
\end{example}

\section{Incomplete Parity Games}\label{section:incomplete_games}

In many practical applications that rely on parity game solving, the parity game is gradually constructed by means of an exploration, often starting from an `initial' vertex.
This is, for instance, the case when using parity games in the context of  model checking or when deciding behavioural preorders or equivalences.
For such applications, it may be profitable to combine exploration and solving, so that the costly exploration can be terminated when the winner of a particular vertex of interest (often the initial vertex) has been determined.
The example below, however, illustrates that one cannot naively solve the parity game constructed so far.

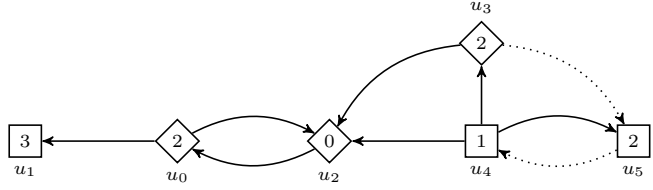
\begin{figure}[h]
\centering
\scriptsize
\begin{tikzpicture}[->,>=stealth',shorten >=0pt,auto,node distance=2.0cm,semithick]
\tikzstyle{even}=[draw,inner sep=2pt,diamond]
\tikzstyle{odd}=[draw,inner sep=2pt,regular polygon,regular polygon sides=4]	
\node[even,label=below:$u_0$] (s1)  {2};
\node[odd,label=below:$u_1$,left of=s1] (s2) {3};
\node[even,label=below:$u_2$,right of=s1] (s2) {0};
\node[odd,label=below:$u_4$,right of=s2]  (s4)  {1};
\node[odd,label=below:$u_5$,right of=s4]  (s5)  {2};
\node[even,label=above:$u_3$,above of=s4,yshift=-20pt]  (s3)  {2};
\path
        (s1) edge (s0)
	(s1) edge[bend left=28] (s2)
	(s2) edge[bend left=28] (s1)
 	(s3) edge[bend right=28] (s2)
 	(s3) edge[dotted,bend left=28] (s5)
 	(s4) edge (s3)
 	(s4) edge (s2)
 	(s4) edge[bend left=28] (s5)
 	(s5) edge[dotted,bend left=28] (s4)
;
\end{tikzpicture}
\caption{A parity game where the dotted edges are not yet known.}\label{fig:partial_parity_game}
\end{figure}

\begin{example}\label{example:incomplete_game}
	Consider the parity game $G$ in Figure~\ref{fig:partial_parity_game}, consisting of all vertices and only the solid edges.
  This game could, for example, be the result of an exploration starting from $u_4$.
  Then $G \cap \{u_0, u_1, u_2, u_3, u_4, u_5\}$ is a subgame for which we can conclude that all vertices form an \peven-dominion.
  However, after exploring the dotted edges, player $\podd$ can escape to vertex $u_4$ from vertex $u_5$.
  Consequently, vertices $u_4$ and $u_5$ are no longer won by player $\peven$ in the extended game.
  Furthermore, observe that the additional edge from $u_3$ to $u_5$ does not affect the previously established fact that player $\peven$ wins this vertex. \qed
\end{example}
To facilitate reasoning about games with incomplete information, we first introduce the notion of an \emph{incomplete} parity game.

\begin{definition}
  An \emph{incomplete} parity game is a structure $\Game = (G, I)$, where $G$ is a parity game $(V, \edgerel, \priof, (V_\peven, V_\podd))$, and $I \subseteq V$ is a set of vertices with potentially unexplored successors. We refer to the set $I$ as the set of \emph{incomplete vertices}; the set  $V \setminus I$ is the set of \emph{complete vertices}.
\end{definition}

Observe that $(G,\emptyset)$ is a `standard' parity game.
We permit ourselves to use the notation for parity game notions such as plays, strategies, dominions, \emph{etcetera} also in the context of incomplete parity games.
In particular, for $\Game = (G,I)$, we will write $\pre{\Game}{U}$ and $\attr{\player}{\Game}{U}$ to indicate $\pre{G}{U}$ and $\attr{\player}{G}{U}$, respectively.
Furthermore, we define $\Game \cap U$ as the structure $(G \cap U, I \cap U)$.

Intuitively, while exploring a parity game, we extend the set of vertices and edges by exploring the incomplete vertices.
Doing so gives rise to potentially new incomplete vertices.
At each stage in the exploration, the incomplete parity game extends incomplete parity games explored in earlier stages. 
We formalise the relation between incomplete parity games, abstracting from any particular order in which vertices and edges are explored.
\begin{definition}
Let $\Game = ((V, \edgerel, \priof, (V_\peven, V_\podd)), I)$, $\Game' = ((V', \edgerel', \priof', (V'_\peven, V'_\podd)), I')$ be incomplete parity games.
We write $\Game \sqsubseteq \Game'$ iff the following conditions hold:
\begin{enumerate}[label=(\arabic*)]
  \item\label{cond:vertices} $V \subseteq V'$, $V_\peven \subseteq V'_\peven$ and $V_\podd \subseteq V'_\podd$;

  \item\label{cond:edges} $\edgerel \subseteq \edgerel'$ and $((V \setminus I) \times V) \cap \edgerel' \subseteq \edgerel$;
  
  \item\label{cond:prio} $\priof = \priof'\!\!\restriction_{V}$;
  
  \item\label{cond:incomplete} $I' \cap V \subseteq I$
\end{enumerate}
\end{definition}
Conditions~\ref{cond:vertices} and~\ref{cond:prio} are self-explanatory.
Condition~\ref{cond:edges} states that on the one hand, no edges are lost, and, on the other hand, $\edgerel'$ can only add edges from vertices that are incomplete: for complete vertices, $\edgerel'$ specifies no new successors.
Finally, condition~\ref{cond:incomplete} captures that the set of incomplete vertices $I'$ cannot contain vertices that were previously complete.
We note that the ordering $\sqsubseteq$ is reflexive, anti-symmetric and transitive.
\begin{example}\label{example:ordering}
Suppose that $\Game = (G,I)$ is the incomplete parity game depicted in Figure~\ref{fig:partial_parity_game}, where $G$ is the game with all vertices and only the solid edges, and $I = \{u_3,u_5\}$.
Then $\Game \sqsubseteq \Game'$, where $\Game' = (G',I')$ is the incomplete parity game where $G'$ is the depicted game with all vertices and both the solid edges and dotted edges, and $I' = \emptyset$.\qed
\end{example}
Let us briefly return to Example~\ref{example:incomplete_game}.
We concluded that the winner of vertex $u_4$ (and also $u_5$) changed when adding new information.
The reason is that player $\podd$ has a strategy to reach an \emph{incomplete} vertex owned by her.
Such an incomplete vertex may present an opportunity to escape from plays that would be non-winning otherwise.
On the other hand, the incomplete vertex $u_3$ has already been sufficiently explored to allow for concluding that this vertex is won by player $\peven$, even if more successors are added to $u_3$.
This suggests that for some subset of vertices, we can decide their winner in an incomplete parity game and preserve that winner in all future extensions of the game.
We formally characterise this set of vertices in the definition below.
\begin{definition}
Let $\Game = ( G,I)$, with $G =(V, \edgerel, \priof, (V_\peven, V_\podd))$ be an incomplete parity game.
The $\player$-\emph{safe} vertices for $\Game$, denoted by $\safe{\player}{\Game}$, is the set $V \setminus \attr{\otherplayer}{G}{V_{\otherplayer} \cap I}$.
\end{definition}
\begin{example}\label{example:safe}
Consider the incomplete parity game $\Game$ of Example~\ref{example:ordering} once more.
We have $\safe{\peven}{\Game} = \{u_0,u_1,u_2,u_3\}$ and $\safe{\podd}{\Game} = \{u_0,u_1,u_2,u_4,u_5\}$.
\qed
\end{example}

In the remainder of this section, we show that it is indeed the case that while exploring a parity game, one can only safely determine the winners in the sets $\safe{\podd}{\Game}$ and $\safe{\peven}{\Game}$, respectively.
More specifically, we show (Lemma~\ref{lemma:dominion_preserved}) that all $\player$-dominions found in $\safe{\player}{\Game}$ are preserved in extensions of the game, and (Lemma~\ref{lemma:not_closed}) the winner of vertices not in $\safe{\player}{\Game}$ are not necessarily won by the same player in extensions of the game.
\begin{lemma}\label{lemma:dominion_preserved}
  Given two incomplete games $\Game$ and $\Game'$ such that $\Game \sqsubseteq \Game'$.
  Any $\player$-dominion in $\Game \cap \safe{\player}{\Game}$ is also an $\player$-dominion in $\Game'$.
\end{lemma}
\begin{proof}
	Let $\Game = (G,I)$, with $G = (V,\edgerel,\priof, (V_\peven,V_\podd))$, and
	$\Game' = (G',I')$, with $G' = (V',\edgerel',\priof', (V'_\peven,V'_\podd))$.
	Assume that $\Game \sqsubseteq \Game'$ and suppose that $D \subseteq V$ is an $\player$-dominion in $G \cap \safe{\player}{\Game}$.
	Observe that since $V \subseteq V'$, also $D \subseteq V'$.
	This means that player $\player$ must have a strategy $\sigma$ that is closed on $D$ and for which every play consisting with that strategy and starting in a vertex in $D$ is winning for $\player$.
	Let $\sigma$ be such a strategy.
	We define strategy $\sigma'$ as follows.
	Let $\pi = v_0\dots v_n \in V^*V_\player$ be an arbitrary path through $D$.
	Then $\sigma'(\pi) = \sigma(\pi)$ whenever $\sigma$ is defined for $\pi$, and $\sigma'(\pi) = v'$ for an arbitrary $v' \in \{v'' \in D \mid (v_n,v') \in E\} \subseteq \{v'' \in D \mid (v_n,v') \in E' \}$ in case $\{v'' \in D \mid (v_n,v') \in E\} \not= \emptyset$ and $\sigma$ is not defined for $\pi$.
	
	Consider an arbitrary play $v_0\,v_1\,\dots$ in $G'$, consistent with $\sigma'$ and starting in $v_0 \in D$.
	Suppose that there is some vertex on this play that is not in $D$.
	Let $v_i$ be the first such vertex; \ie, $v_j \in D$ for all $j < i$.
	We distinguish two cases.
	\begin{itemize}
	\item Assume $\sigma$ is defined for $v_0\,v_1\,\dots v_{i-1}$. 
	Then $\sigma(v_0\,v_1,\dots v_{i-1}) \in D$, since $\sigma$ is closed on $D$ in $G$.
	But then also $\sigma'(v_0\,v_1\,\dots v_{i-1}) = v_i \in D$. 
	Contradiction.
	
	\item Next, assume that $\sigma$ is not defined for $v_0\,v_1\,\dots v_{i-1}$.
	We distinguish two further cases.
	\begin{itemize}
	\item Case $v_{i-1} \in V_\player$. 
	Since $\sigma$ is closed on $D$ in $G$, it must be the case that $\{v'' \in V \mid (v_{i-1},v'') \in E\} \subseteq D$. 
	By construction, $\sigma'(v_0\,v_1\,\dots v_{i-1}) = v_i \in D$.
	Contradiction.
	\item Case $v_{i-1} \in V_{\otherplayer}$. 
	Since $\sigma$ is closed on $D$ in $G$, it must be the case that $\{v '' \in V \mid (v_{i-1},v'') \in E\} \subseteq D$.
	Furthermore, since $D \subseteq V \cap \safe{\player}{\Game}$,  $(I' \cap V) \subseteq I$ and $(V \setminus I) \times V) \cap E' \subseteq E$, also $\{v '' \in V' \mid (v_{i-1},v'') \in E'\} \subseteq D$.
	But then also $v_i \in D$.
	Contradiction.
	
	\end{itemize}
	\end{itemize}
	Since all cases lead to a contradiction, we find that $v_i \in D$.
	But then all plays consistent with $\sigma'$ must remain in $D$.
	
	It remains to argue that $\sigma'$ is a winning play.
	However, this follows from the fact that all plays $\pi$, starting in $D$ and consistent with $\sigma'$ are also plays consistent with $\sigma$.
	Since $\sigma$ was a winning strategy for player $\player$, also $\sigma'$ is winning for $\player$ in $(G',I')$.
	\qed
\end{proof}

\begin{example}
Recall that in Example~\ref{example:safe}, we found that $\safe{\peven}{\Game} = \{u_0,u_1,u_2,u_3\}$.
Observe that in the incomplete parity game $\Game$ of Example~\ref{example:ordering}, restricted to vertices $\{u_0,u_1,u_2,u_3\}$, all vertices are won by player $\peven$, and, hence, $\{u_0,u_1,u_2,u_3\}$ is an $\peven$-dominion.
Following Lemma~\ref{lemma:dominion_preserved} we can indeed conclude that this remains an $\peven$-dominion in all extensions of $\Game$, and, in particular, for the (complete) parity game $\Game'$ of Example~\ref{example:ordering}. \qed
\end{example}

\begin{lemma}\label{lemma:not_closed}
  Let $\Game$ be an incomplete parity game.
  Suppose that $W$ is an $\player$-dominion in $\Game$.
  If $W \not\subseteq \safe{\player}{\Game}$, then there is an (incomplete) parity game $\Game'$ such that $\Game \sqsubseteq \Game'$
 and all vertices in $W \setminus \safe{\player}{\Game}$ are won by $\otherplayer$.
\end{lemma}
\begin{proof}
	Let $\Game = (G,I)$ be an incomplete parity game, with $G = (V,\edgerel,\priof,(V_\peven,V_\podd))$, and assume that $W \subseteq V$ is an $\player$-dominion.
	Suppose that $W \not\subseteq \safe{\player}{\Game}$.
	
	Let $G' = (V',\edgerel',\priof',(V_\peven,V_\podd))$, with $V' = V \cup \{z\}$, for fresh vertex $z \notin V$, $E' = E \cup \{(v,z) \mid v \in I\}$, $\priof'\!\!\restriction_{V} = \priof$ and $\priof'(z) = 0$, and $V'_\player = V_\player \cup \{z\}$ and $V'_{\otherplayer} = V_{\otherplayer}$.
	Then it follows that $\Game \sqsubseteq \Game'$ for $\Game' = (G', \emptyset)$.
	
	Pick a vertex $v \in W \setminus \safe{\player}{\Game}$. 
	Then $v \in \attr{\otherplayer}{G}{V_{\otherplayer}\cap I}$, so player $\otherplayer$ must have a strategy to force play to $V_{\otherplayer}\cap I$.
	Let $\sigma_{\otherplayer}$ be this strategy.
	We define a new strategy $\sigma'_{\otherplayer}$ which, for sequences $v_0\dots v_n$ for which $v_n \notin V_{\otherplayer} \cap I$ is defined as $\sigma_{\otherplayer}$, and for $v_n \in V_{\otherplayer} \cap I$, we have $\sigma'_{\otherplayer}(v_0\dots v_n) = z$.
	Let $\sigma_{\player}$ be an arbitrary strategy for player $\player$.
	Consider a play $\pi$, starting in $v$, which is consistent with both $\sigma_{\player}$ and $\sigma'_{\otherplayer}$.
	Then there must be a vertex $w$ on that play such that $w \in I$, since $\sigma'_{\otherplayer}$ forces play to $V_{\otherplayer} \cap I$.
	But then, by construction, $\pi$ must be finite and end in sink $z \notin W$.
	Note that vertex $z \in V'$ is won by $\otherplayer$, so we find that $\pi$ is won by $\otherplayer$.
	Since $\sigma_{\player}$ is an arbitrary strategy by player $\player$, we can conclude that all plays consistent with $\sigma'_{\otherplayer}$ are won by $\otherplayer$ and therefore vertex $v$ is won by $\otherplayer$, too.
	Hence, all vertices in $W \setminus \safe{\player}{\Game}$ are won by $\otherplayer$ in $\Game'$. \qed
\end{proof}

As a corollary of the above lemma, we find that $\player$-dominions that contain vertices outside of the $\player$-safe set are not guaranteed to be dominions in all extensions of the incomplete parity game.
\begin{corollary}\label{cor:not_closed}
  Let $\Game$ be an incomplete parity game.
  Suppose that $W$ is an $\player$-dominion in $\Game$.
  If $W \not\subseteq \safe{\player}{\Game}$, then there is an (incomplete) parity game $\Game'$ such that $\Game \sqsubseteq \Game'$ and $W$ is \emph{not} an $\player$-dominion in $\Game'$.
\end{corollary}

The theorem below summarises the two previous results, claiming that the sets $\safe{\peven}{\Game}$ and $\safe{\podd}{\Game}$ are the optimal subsets that can be used safely when combining solving and the exploration of a parity game.

\begin{theorem}\label{th:optimality}
  Let $\Game = (G, I)$, with $G = (V, \edgerel, \priof, (V_\peven,V_\podd))$, be an incomplete parity game.
  Define $W_\player$ as the union of all $\player$-dominions in $\Game \cap \safe{\player}{\Game}$, and let $W_? = V \setminus (W_\peven \cup W_\podd)$.
  Then $W_?$ is the largest set of vertices $v$ for which there are incomplete parity games $\Game^\player$ and $\Game^{\otherplayer}$ such that $\Game \sqsubseteq \Game^\player$ and $\Game \sqsubseteq \Game^{\otherplayer}$ and $v$ is won by $\player$ in $\Game^\player$ and $v$ is won by $\otherplayer$ in $\Game^{\otherplayer}$.
\end{theorem}
\begin{proof}
Let $\Game$, with $G = (V, \edgerel, \priof, (V_\peven, V_\podd))$ be an incomplete parity game.
Pick a vertex $v \in W_?$.
Suppose that in $G$, vertex $v \in W_?$ is won by player $\player$.
Let $\Game^\player = \Game$. 
Then $\Game \sqsubseteq \Game^\player$ and $v$ is also won by $\player$ in $\Game^\player$.

Next, we argue that there must be a game $\Game^{\otherplayer}$ such that $\Game \sqsubseteq \Game^{\otherplayer}$ and $v$ is won by $\otherplayer$ in $\Game^{\otherplayer}$.
Since $v \in W_?$ is won by player $\player$ in $G$, $v$ must belong to an $\player$-dominion in $G$.
Towards a contradiction, assume that $v \in \safe{\player}{\Game}$. 
Then there must also be a $\player$-dominion containing $v$ in $G\cap \safe{\player}{\Game}$, since $\otherplayer$ cannot escape the set $\safe{\player}{\Game}$. 
But then $v \in W_\player$.
Contradiction, so $v \notin \safe{\player}{\Game}$.
So, $v$ must be part of an $\player$-dominion $D$ in $G$ such that $D \not\subseteq \safe{\player}{\Game}$.
By Lemma~\ref{lemma:not_closed}, we find that there is an incomplete parity game $\Game^{\otherplayer}$ such that $\Game \sqsubseteq \Game^{\otherplayer}$ and all vertices in $D \setminus \safe{\player}{\Game}$, and vertex $v \in D$ in particular, are won by $\otherplayer$ in $\Game^{\otherplayer}$.

Finally, we argue that $W_?$ cannot be larger.
Pick a vertex $v \notin W_?$.
Then there must be some player $\player$ such that $v \in W_\player$, and, consequently, there must be an $\player$-dominion $D \subseteq \Game \cap \safe{\player}{\Game}$ such that $v \in D$.
But then by Lemma~\ref{lemma:dominion_preserved}, we find that $v$ is won by $\player$ in all incomplete parity games $\Game'$ such that $\Game \sqsubseteq \Game'$.\qed
\end{proof}

\section{On-the-fly Solving}\label{sec:partial_solvers}

In the previous section we saw that for any solver $\textsf{solve}_\player$, which accepts a parity game as input and returns an $\player$-dominion $W_\player$, a correct on-the-fly solving algorithm can be obtained by computing $W_\player = \textsf{solve}_\player(\Game \cap \safe{\player}{\Game})$ while exploring an (incomplete) parity game $\Game$.
While this approach is clearly sound, computing the set of safe vertices can be expensive for large state spaces and potentially wasteful when no dominions are found afterwards.
We next introduce \emph{safe attractors} which, we show, can be used to search for specific dominions without first computing the $\player$-safe set of vertices.

\subsection{Safe Attractors}

We start by observing that the $\player$-attractor to a set $U$ in an incomplete parity game $\Game$ does not make a distinction between the set of complete and incomplete vertices.
Consequently, it may wrongly conclude that $\player$ has a strategy to force play to $U$ when the attractor strategy involves incomplete vertices owned by $\otherplayer$.
We thus need to make sure that such vertices are excluded from consideration.
This can be achieved by considering the set of \emph{unsafe} vertices $V_{\otherplayer} \cap I$ as potential vertices that can be used by the other player to escape.
We define the safe $\player$-attractor as the least fixed point of the \emph{safe} control predecessor.
The latter is defined as follows:
\begin{equation*}
  \spre{\player}{\Game}{U} = (V_{\player} \cap \pre{\Game}{U}) \cup (V_{\otherplayer} \setminus (\pre{\Game}{V \setminus U} \cup \sinks{\Game} \cup I))
\end{equation*}
\begin{lemma}\label{lemma:safe_predecessor}
  Let $\Game$ be an incomplete parity game.
  For all vertex sets $X \subseteq \safe{\player}{\Game}$ it holds that $\cpre{\player}{\Game \cap \safe{\player}{\Game}}{X} = \spre{\player}{\Game}{X}$.
\end{lemma}

\begin{proof}  
  Let $\Game = (G,I)$, with $G = (V, \edgerel, \priof, (V_\peven, V_\podd))$, be an incomplete parity game and let $\Game' = \Game \cap \safe{\player}{\Game} = (G', I \cap \safe{\player}{\Game})$, with $G' = (V', \edgerel', \priof', (V'_\peven, V'_\podd))$.
  Pick an arbitrary vertex set $X \subseteq \safe{\player}{\Game}$.

  \begin{itemize}
  \item \emph{ad} $\cpre{\player}{\Game'}{X} \subseteq \spre{\player}{\Game}{X}$.
   Let $v \in \cpre{\player}{\Game'}{X} \subseteq \safe{\player}{\Game}$.
   We distinguish two cases: $v \in V'_\player$ and $v \in V'_{\otherplayer}$.
   \begin{itemize}

   \item Case $v \in V'_\player$.
   Then, by definition of the control predecessor, also $v \in V'_\player \cap \pre{\Game'}{X}$.
   Since $V'_\player \subseteq V_\player$, it suffices to show that $v \in \pre{\Game}{X}$.
   But this follows instantly since $\Game'$ is a substructure of $\Game$.
   Hence, $v \in V_\player \cap \pre{\Game'}{X} \subseteq \spre{\player}{\Game}{X}$.

   \item Case $v \in V'_{\otherplayer}$.
   Since $v \in \cpre{\player}{\Game'}{X}$, we then also have $v \in V'_{\otherplayer} \setminus (\pre{\Game'}{V' \setminus X} \cup \sinks{\Game'})$, so $v \notin \pre{\Game'}{V' \setminus X} \cup \sinks{\Game'}$.

   Suppose $v \in I$.
   Since $v \in V'_{\otherplayer} \subseteq V_{\otherplayer}$, also $v \in V_{\otherplayer} \cap I$, and therefore $v \in \attr{\otherplayer}{\Game}{V_{\otherplayer} \cap I}$.
   But this contradicts $v \in \safe{\player}{\Game} = V \setminus \attr{\otherplayer}{\Game}{V_{\otherplayer} \cap I}$.
   Hence, $v \notin I$.
   Next, suppose that $v \in \sinks{\Game}$.
   Since $v \notin \sinks{\Game'}$, vertex $v$ was removed from $\Game'$, and, hence, $v \notin \safe{\player}{\Game}$.
   Contradiction, so $v \notin \sinks{\Game}$.
   Finally, suppose that $v \in \pre{\Game}{V \setminus X}$.
   Then there must be a vertex $w \in vE$ for which $w \in V \setminus X$ and $w \not\in V' \setminus X$.
   However, then $w \notin \safe{\player}{\Game}$ and also $v \notin \safe{\player}{\Game}$ since $v \in V_{\otherplayer}$.
   So, $v \notin \pre{\Game}{V \setminus X}$.
   Hence, we can conclude that $v \in V_{\otherplayer}\setminus (\pre{\Game}{V \setminus X} \cup \sinks{\Game} \cup I)$.

   \end{itemize}
   We may thus conclude that $v \in \spre{\player}{\Game}{X}$.

  \item \emph{ad} $\cpre{\player}{\Game'}{X} \supseteq \spre{\player}{\Game}{X}$.
   Let $v \in \spre{\player}{\Game}{X}$.
   We again use a case distinction to show that $v \in \cpre{\player}{\Game'}{X}$. 
   \begin{itemize}
   \item Case $v \in V_\player$.
    Then also $v \in \pre{\Game}{X}$, so for some $w \in vE$, we have $w \in X \subseteq \safe{\player}{\Game}$.
    But then also $v \in \safe{\player}{\Game}$, since $\player$ can move to a vertex in $X$, preventing plays passing through $v$ from reaching $V_{\otherplayer} \cap I$.
    Hence, $v \in V_\player \cap \safe{\player}{\Game} = V'_\player$.
    Since $X \subseteq \safe{\player}{\Game}$, we also have $v \in \pre{\Game'}{X}$.
    Hence, $v \in V'_\player \cap \pre{\Game'}{X}$.

   \item Case $v \in V_{\otherplayer}$.
    Then $v \notin \pre{\Game}{V\setminus X} \cup \sinks{\Game} \cup I$.
    Assume that $v \in \pre{\Game'}{V'\setminus X}$.
    Since $\Game'$ is a subgraph of $\Game$ and $V' \setminus X \subseteq V \setminus X$, we then have $v \in \pre{\Game}{V \setminus X}$.
    Contradiction.
    Next, assume that $v \in \sinks{\Game'}$.
    But this contradicts $v \notin \sinks{\Game}$, since $\sinks{\Game'} \subseteq \sinks{\Game}$.
    Finally, we know that $v \notin \pre{\Game}{V \setminus X}$.
    This means that all successors of $v$ must be in $X$, \ie, we find that $v E \subseteq X \subseteq \safe{\player}{\Game}$.
    But then, since $v \notin \sinks{\Game} \cup I$, we have $v \notin \attr{\otherplayer}{\Game}{V_{\otherplayer} \cap I}$, so $v \in \safe{\player}{\Game}$.
    Hence, $v \in V'_\player$.
   \end{itemize}
   We may therefore conclude that $v \in \cpre{\player}{\Game'}{X}$.\qed
  \end{itemize}
\end{proof}

The safe $\player$-attractor to $U$, denoted $\sattr{\player}{\Game}{U}$, is the set of vertices from which player $\player$ can force to \emph{safely} reach $U$ in $\Game$:
\begin{equation*}
\sattr{\player}{\Game}{U} = \mu Z.(U \cup \spre{\player}{\Game}{Z})
\end{equation*}
\begin{lemma}
Let $\Game$ be an incomplete parity game, and $X \subseteq \safe{\player}{\Game}$.
Then $\attr{\player}{\Game \cap \safe{\player}{\Game}}{X} = \sattr{\player}{\Game}{X}$.
\end{lemma}
\begin{proof}
  We show by means of an induction that the fixed point approximants $A_i$ of $\attr{\player}{\Game \cap \safe{\player}{\Game}}{X}$ are equal to the approximants $B_i$ of $\sattr{\player}{\Game}{X}$ and that $A_i \subseteq \safe{\player}{\Game}$.
  Initially, $A_0 = B_0 = \emptyset$, which are equal and $\emptyset \subseteq \safe{\player}{\Game}$.
  
  Inductive step, assume that $A_i = B_i$ and $A_i \subseteq \safe{\player}{\Game}$.
  We can show that $X \cup \cpre{\player}{G \cap \safe{\player}{\Game}}{A_i} = X \cup \spre{\player}{\Game}{B_i}$.
  First, using Lemma~\ref{lemma:safe_predecessor} we can conclude that $\cpre{\player}{\Game \cap \safe{\player}{\Game}}{A_i} = \spre{\player}{\Game}{A_i}$.
  From $A_i = B_i$ it follows that $X \cup \cpre{\player}{G \cap \safe{\player}{\Game}}{A_i} = X \cup \spre{\player}{\Game}{B_i}$.
  Finally, $\cpre{\player}{G \cap \safe{\player}{\Game}}{A_i} \subseteq \safe{\player}{\Game}$ by definition.
  \qed
\end{proof}

In particular, we can conclude the following:
\begin{corollary}\label{col:safe_attractor}
  Let $\Game$ be an incomplete parity game, and $X \subseteq \safe{\player}{\Game}$ be an $\player$-dominion.
  Then $\sattr{\player}{\Game}{X}$ is an $\player$-dominion for all $\Game'$ satisfying $\Game \sqsubseteq \Game'$.
\end{corollary}
One application of the above corollary is the following: since on-the-fly solving is typically performed repeatedly, previously found dominions can be expanded by computing the safe $\player$-attractor towards these already solved vertices.
Another corollary is the following, which states that complete sinks can be safely attracted towards. 
\begin{corollary}
  Let $\Game = (G,I)$ be an incomplete parity game and let $\Game'$ be such that $\Game \sqsubseteq \Game'$.
  Then $\sattr{\player}{\Game}{\psinks{\otherplayer}{\Game}\setminus I}$ is an $\player$-dominion in $\Game'$.
\end{corollary}

\subsection{Partial Solvers}

In practice, a full-fledged solver, such as Zielonka's algorithm~\cite{Zielonka98} or one of the Priority Promotion variants~\cite{BenerecettiDM18}, may be costly to run often while exploring a parity game.
Instead, cheaper partial solvers may be used that search for a dominion of a particular shape.
We study three such partial solvers in this section, with a particular focus on solvers that lend themselves for parity games that are represented symbolically using, \eg, BDDs~\cite{Bryant92}, MDDs~\cite{Miller93} or LDDs~\cite{DijkP17}.
For the remainder of this section, we fix an arbitrary incomplete parity game $\Game = ( (V,\edgerel,\priof,(V_\peven,V_\podd)),I)$. 

\paragraph{Winning solitaire cycles.}

A simple cycle in  $\Game$ can be represented by a finite sequence of distinct vertices $v_0\,  v_1\, \ldots\, v_n$ satisfying $v_0 \in v_nE$.
Such a cycle is an \emph{$\player$-solitaire cycle} whenever all vertices on that cycle are owned by player $\player$.

Observe that if all vertices on an $\player$-solitaire cycle have a priority that is of the same parity as the owner $\player$, then all vertices on that cycle are won by player $\player$.
Formally, these are thus cycles through vertices in the set $P_\player \cap V_\player$, where $P_\peven = \{ v \in V \setminus \sinks{\Game} \mid \priof(v)\mod 2 = 0\}$ and $P_\podd = \{ v \in V \setminus \sinks{\Game} \mid \priof(v)\mod 2 = 1\}$.
Let $\solcycle{\player}{\Game}$ represent the largest set of $\player$-solitaire winning cycles.
Then $\solcycle{\player}{\Game}= \nu Z. ( P_\player \cap V_\player \cap \pre{\Game}{Z})$.
\begin{proposition}\label{prop:solcycle}
  The set $\solcycle{\player}{\Game}$ is an $\player$-dominion and we have $\solcycle{\player}{\Game} \subseteq \safe{\player}{\Game}$.
\end{proposition}
\begin{proof} 
  We first prove that $\solcycle{\player}{\Game} \subseteq \safe{\player}{\Game}$.
  We show, by means of an induction on the fixed point approximants $A_i$ of the attractor, that $\solcycle{\player}{\Game} \cap  \attr{\otherplayer}{\Game}{V_{\otherplayer} \cap I} = \emptyset$.
  The base case follows immediately, as $\solcycle{\player}{\Game} \cap A_0 = \solcycle{\player}{\Game} \cap \emptyset = \emptyset$.
  For the induction, we assume that $\solcycle{\player}{\Game} \cap A_i = \emptyset$; we show that also $\solcycle{\player}{\Game} \cap ((V_{\otherplayer} \cap I) \cup \cpre{\otherplayer}{\Game}{A_i}) = \emptyset$.
  First, observe that $\solcycle{\player}{\Game} \subseteq V_\player$; hence, it suffices to prove that $\solcycle{\player}{\Game} \cap (V_\player \setminus (\pre{\Game}{V \setminus A_i} \cup \sinks{\Game}) =\emptyset$.
  But this follows immediately from the fact that for every vertex $v \in \solcycle{\player}{\Game}$, we have $v \in P_\player \cap V_\player \cap \pre{\Game}{\solcycle{\player}{\Game}}$; more specifically, we have $vE \cap \solcycle{\player}{\Game} \not= \emptyset$ for all $v \in \solcycle{\player}{\Game}$.
  
  The fact that $\solcycle{\player}{\Game}$ is an $\player$-dominion follows from the fact that for every vertex $v \in \solcycle{\player}{\Game}$, there is some $w \in vE \cap \solcycle{\player}{\Game}$.
  This means that player $\player$ must have a strategy that is closed on $\solcycle{\player}{\Game}$.
  Since all vertices in $\solcycle{\player}{\Game}$ are of the priority that is beneficial to $\player$, this closed strategy is also winning for $\player$.\qed
\end{proof}
Observe that winning solitaire cycles can be computed without first computing the $\player$-safe set.
Parity games that stand to profit from detecting winning solitaire cycles are those originating from verifying safety properties.

\paragraph{Winning forced cycles.}
In general, a cycle in $\safe{\player}{\Game}$, through vertices in $P_\peven$ can contain vertices of both players, providing player $\podd$ an opportunity to break the cycle if that is beneficial to her.
Nevertheless, if breaking a cycle always inadvertently leads to another cycle through $P_\peven$, then we may conclude that all vertices on these cycles are won by player $\peven$.
We call these cycles \emph{winning forced cycles} for player $\peven$.
A dual argument applies to cycles through $P_\podd$.
Let $\allcycle{\player}{\Game}$ represent the largest set of vertices that are on winning forced cycles for player $\player$.
More formally, we define $\allcycle{\player}{\Game} = \nu Z. ( P_\player \cap \safe{\player}{\Game} \cap \cpre{\player}{\Game}{Z})$.
\begin{lemma}\label{lemma:allcycle}
  The set $\allcycle{\player}{\Game}$ is an $\player$-dominion and we have $\allcycle{\player}{\Game} \subseteq \safe{\player}{\Game}$.
\end{lemma}
\begin{proof}
  The fact that $\allcycle{\player}{\Game} \subseteq \safe{\player}{\Game}$ follows immediately from the fact that for all $v \in \allcycle{\player}{\Game}$, we have $v \in P_\player \cap \safe{\player}{\Game} \cap \cpre{\player}{\Game}{\allcycle{\player}{\Game}}$.
  
  We next show that $\allcycle{\player}{\Game}$ is an $\player$-dominion.
  Pick an arbitrary vertex $v \in \allcycle{\player}{\Game}$.
  If $v \in V_\player$, then $vE \cap \allcycle{\player}{\Game} \not= \emptyset$, so player $\player$ has a strategy to move to another vertex in $\allcycle{\player}{\Game}$.
  In case $v \in V_{\otherplayer}$, then $vE \subseteq \allcycle{\player}{\Game}$, so any play passing through $v$ is guaranteed to next visit a vertex in $\allcycle{\player}{\Game}$.
  Hence, player $\player$ has a closed strategy on $\allcycle{\player}{\Game}$.
  Since $\allcycle{\player}{\Game} \subseteq P_\player$, such a closed strategy must be winning for $\player$. \qed
\end{proof}

A possible downside of the above construction is that it again requires to first compute $\safe{\player}{\Game}$, which, in particular cases, may incur an additional overhead.
Instead, we can compute the same set using the safe control predecessor.
We define $\sallcycle{\player}{\Game} = \nu Z. ( P_\player \cap \spre{\player}{\Game}{Z})$.
\begin{proposition}\label{prop:forced_cycles}
	We have $\allcycle{\player}{\Game} = \sallcycle{\player}{\Game}$.
\end{proposition}
\begin{proof}
	Let $\tau(Z) = P_\player \cap \spre{\player}{\Game}{Z}$.
	We use set inclusion to show that $\allcycle{\player}{\Game}$ is indeed a fixed point of $\tau$.
	\begin{itemize}
		\item \emph{ad} $\allcycle{\player}{\Game} \subseteq \tau(\allcycle{\player}{\Game})$.
		Pick a vertex $v \in \allcycle{\player}{\Game}$.
		By definition of $\allcycle{\player}{\Game}$, we have $v \in P_\player \cap \safe{\player}{\Game} \cap \cpre{\player}{\Game}{\allcycle{\player}{\Game}}$.
		Observe that $\safe{\player}{\Game} \cap \cpre{\player}{\Game}{\allcycle{\player}{\Game}} = \safe{\player}{\Game} \cap \cpre{\player}{\Game \cap \safe{\player}{\Game}}{\allcycle{\player}{\Game}}$.
		But then, since $\allcycle{\player}{\Game} \subseteq \safe{\player}{\Game}$, we find, by Lemma~\ref{lemma:safe_predecessor}, that $\cpre{\player}{\Game \cap \safe{\player}{\Game}}{\allcycle{\player}{\Game}} = \spre{\player}{\Game}{\allcycle{\player}{\Game}}$.
		Hence, $v \in P_\player \cap \spre{\player}{\Game}{\allcycle{\player}{\Game}} = \tau(\allcycle{\player}{\Game})$.
		
		\item \emph{ad}  $\allcycle{\player}{\Game} \supseteq \tau(\allcycle{\player}{\Game})$.
		Again pick a vertex $v \in \tau(\allcycle{\player}{\Game})$.
		Then $v \in P_\player \cap \spre{\player}{\Game}{\allcycle{\player}{\Game}}$.
		Since $\allcycle{\player}{\Game} \subseteq \safe{\player}{\Game}$, by Lemma~\ref{lemma:safe_predecessor}, we again have $\spre{\player}{\Game}{\allcycle{\player}{\Game}} = \cpre{\player}{\Game \cap \safe{\player}{\Game}}{\allcycle{\player}{\Game}}$.
		But then it must be the case that $v \in \safe{\player}{\Game}$.
		Moreover, $\cpre{\player}{\Game \cap \safe{\player}{\Game}}{\allcycle{\player}{\Game}} \subseteq \cpre{\player}{\Game}{\allcycle{\player}{\Game}}$.
		So $v \in P_\player \cap \safe{\player}{\Game} \cap \cpre{\player}{\Game}{\allcycle{\player}{\Game}} = \allcycle{\player}{\Game}$.

	\end{itemize}
	We show next that for any $Z = \tau(Z)$, we have $Z \subseteq \allcycle{\player}{\Game}$.
	Let $Z$ be such.
	We first show that for every $v \in Z \cap V_\player$, there is some $w \in vE \cap Z$, and for every $v \in Z \cap V_{\otherplayer}$, we have $v \notin \sinks{\Game}$, $v \notin I$ and $v E \subseteq Z$.
	Pick $v \in Z \cap V_\player$.
	Then $v \in \tau(Z) \cap V_\player = P_\player \cap V_\player \cap \spre{\player}{\Game}{Z} \subseteq \pre{\Game}{Z}$.
	But then $vE \cap Z \not= \emptyset$.
	Next, let $v \in Z \cap V_{\otherplayer}$.
	Then $v \in \tau(Z) \cap V_{\otherplayer} = P_\player \cap V_{\otherplayer} \cap \spre{\player}{\Game}{Z} \subseteq V_{\otherplayer} \setminus (\pre{\Game}{V \setminus Z} \cup \sinks{\Game} \cup I)$.
	So $v \notin \pre{\Game}{V \setminus Z} \cup \sinks{\Game} \cup I$.
	Consequently, $vE \subseteq Z$, $v \notin \sinks{\Game}$ and $v \notin I$.
	
	Since for every $v \in Z \cap V_\player$, we have $vE \cap Z \neq \emptyset$, there must be a strategy for player $\player$ to move to another vertex in $Z$.
	Let $\sigma$ be this strategy.
	Moreover, since for all $v \in Z \cap V_{\otherplayer}$ we have $vE \subseteq Z$, we find that $\sigma$ is closed on $Z$ and since $Z \cap \sinks{\Game} = \emptyset$, strategy $\sigma$ induces forced cycles.
	Moreover, since $Z \subseteq P_\player$, we can conclude that all vertices in $Z$ are on winning forced cycles.
	
	Finally, we must argue that $Z \subseteq \safe{\player}{\Game}$. 
	But this follows from the fact that $Z \cap V_{\otherplayer} \cap I = \emptyset$, and, hence, also $Z \cap \attr{\otherplayer}{\Game}{V_{\otherplayer} \cap I}  = \emptyset$.
Since $Z$ is contained within $P_\player \cap \safe{\player}{\Game}$, we find that $Z \subseteq \allcycle{\player}{\Game}$.
\qed
\end{proof}

\paragraph{Fatal attractors.}
Both solitaire cycles and forced cycles utilise the fact that the parity winning condition becomes trivial if the only priorities that occur on a play are of the parity of a single player.
Fatal attractors~\cite{HuthKP13} were originally conceived to solve parts of a game using algorithms that have an appealing worst-case running time; for a detailed account, we refer to~\cite{HuthKP13}.
While \emph{ibid.} investigates several variants, the main idea behind a fatal attractor is that it identifies cycles in which the priorities are non-decreasing until the dominating priority of the attractor is (re)visited.
We focus on a simplified (and cheaper) variant of the \texttt{psolB} algorithm of~\cite{HuthKP13}, which is based on the concept of a \emph{monotone} attractor, which, in turn, relies on the monotone control predecessor defined below, where $P^{\geq c} = \{ v \in V \mid \priof(v) \geq c \}$:
\begin{align*}
  \Mpre{\player}{\Game}{Z}{U}{c} &= P^{\geq c} \cap \cpre{\player}{\Game}{Z \cup U}
\end{align*}
The monotone attractor for a given priority is then defined as the least fixed point of the monotone control predecessor for that priority, formally $\Mattr{\player}{\Game}{U}{c} = \mu Z. \Mpre{\player}{\Game}{Z}{U}{c}$.
A \emph{fatal} attractor for priority $c$ is then the largest set of vertices closed under the monotone attractor for priority $c$; \ie, $\fattr{\player}{\Game}{c} = \nu Z. (P^{=c} \cap \safe{\player}{\Game} \cap \Mattr{\player}{\Game \cap \safe{\player}{\Game}}{Z}{c})$, where $P^{=c} = P^{\geq c} \setminus P^{\geq c+1}$.
\begin{lemma}[See~\cite{HuthKP13}, Theorem 2]
	For even $c$, we have that $\Mattr{\peven}{\Game \cap \safe{\player}{\Game}}{\fattr{\peven}{\Game}{c}}{c} \subseteq \safe{\peven}{\Game}$ and $\Mattr{\peven}{\Game \cap \safe{\player}{\Game}}{\fattr{\peven}{\Game}{c}}{c}$ is an $\peven$-dominion.
	If $c$ is odd then we have $\Mattr{\podd}{\Game \cap \safe{\player}{\Game}}{\fattr{\podd}{\Game}{c}}{c} \subseteq \safe{\podd}{\Game}$ and $\Mattr{\podd}{\Game \cap \safe{\player}{\Game}}{\fattr{\podd}{\Game}{c}}{c}$ is an $\podd$-dominion.
\end{lemma}
Our simplified version of the \texttt{psolB} algorithm, here dubbed \texttt{solB$^-$} computes fatal attractors for all priorities in descending order, accumulating $\peven$ and $\podd$-dominions and extending these dominions using a standard $\peven$ or $\podd$-attractor.
This can be implemented using a simple loop over these priorities.

In line with the previous solvers, we can also modify this solver to employ a safe monotone control predecessor, which uses a construction that is similar in spirit to that of the safe control predecessor.
Formally, we define the safe monotone control predecessor as follows:
\begin{align*}
  \sMpre{\player}{\Game}{Z}{U}{c} &= P^{\geq c} \cap \spre{\player}{\Game}{Z \cup U}
\end{align*}
The corresponding safe monotone $\alpha$-attractor, denoted $\sMattr{\player}{\Game}{U}{c}$, is defined as follows: $\sMattr{\player}{\Game}{U}{c} = \mu Z. \sMpre{\player}{\Game}{Z}{U}{c}$.
We define the \emph{safe} fatal attractor for priority $c$ as the set $\sfattr{\player}{\Game}{c} =\nu Z. (P^{=c} \cap \sMattr{\player}{\Game}{Z}{c})$.
Similar to the safe attractor case using a standard inductive argument we can show the following.
\begin{lemma}\label{lemma:safe_monotone_attractor}
  Let $\Game$ be an incomplete parity game, and $X \subseteq \safe{\player}{\Game}$.
  Then it holds that $\Mattr{\player}{\Game \cap \safe{\player}{\Game}}{X}{c} = \sMattr{\player}{\Game}{X}{c}$ for any player $\player$ and priority $c$.
\end{lemma}
\begin{proof}
  We show by means of an induction that the fixed point approximants $A_i$ of $\Mattr{\player}{\Game \cap  \safe{\player}{\Game}}{X}{c}$ are equal to the approximants $B_i$ of $\sMattr{\player}{\Game}{X}{c}$ and $A_i \subseteq \safe{\player}{\Game}$.
  Initially, $A_0 = B_0 = \emptyset$ and $\emptyset \subseteq \safe{\player}{\Game}$.
  
  Inductive step, assume that $A_i = B_i$ and $A_i \subseteq \safe{\player}{\Game}$.
  We observe that $\Mpre{\player}{\Game \cap \safe{\player}{\Game}}{A_i}{X}{c} = P^{\geq c} \cap \cpre{\player}{\Game \cap \safe{\player}{\Game}}{X \cup A_i}$, which by Lemma~\ref{lemma:safe_predecessor} is equal to $P^{\geq c} \cap \spre{\player}{\Game}{X \cup A_i}$ since $X \cup A_i \subseteq \safe{\player}{\Game}$.
  Therefore, $\Mpre{\player}{\Game \cap \safe{\player}{\Game}}{A_i}{X}{c} = \sMpre{\player}{\Game}{A_i}{X}{c}$.
  Furthermore, $\Mpre{\player}{\Game \cap \safe{\player}{\Game}}{A_i}{X}{c} \subseteq \safe{\player}{\Game}$.
  \qed
\end{proof}

\begin{lemma}\label{lemma:safe_monotone_attractor_safe}
  Let $\Game$ be an incomplete parity game and $X$ a set of vertices such that $X \subseteq \sMattr{\player}{\Game}{X}{c}$.
  Then $\sMattr{\player}{\Game}{X}{c} \subseteq \safe{\player}{\Game}$ for any player $\player$ and priority $c$.
\end{lemma}

\begin{proof}
  Let $\Game = (G,I)$, with $G = (V, \edgerel, \priof, (V_\peven, V_\podd))$, be an incomplete parity game.
	We show, by means of an induction on the fixed point approximants $A_i$ of the attractor, that $\sMattr{\player}{\Game}{X}{c} \cap \attr{\otherplayer}{\Game}{V_{\otherplayer} \cap I} = \emptyset$.	
	The base case follows immediately, as $\sMattr{\player}{\Game}{X}{c} \cap A_0 = \sMattr{\player}{\Game}{X}{c} \cap \emptyset = \emptyset$.
	
  For the induction, we assume that $\sMattr{\player}{\Game}{X}{c} \cap A_i = \emptyset$; we show that also $\sMattr{\player}{\Game}{X}{c} \cap ((V_{\otherplayer} \cap I) \cup \cpre{\otherplayer}{\Game}{A_i}) = \emptyset$.
  First of all, since $\spre{\player}{\Game}{Z} \cap V_{\otherplayer} \cap I = \emptyset$, for any $Z$, it follows that $\sMattr{\player}{\Game}{X}{c} \cap V_{\otherplayer} \cap I = \emptyset$.
  It remains to show that $\sMattr{\player}{\Game}{X}{c} \cap \cpre{\otherplayer}{\Game}{A_i} = \emptyset$, which is equal to showing that $P^{\geq c} \cap \spre{\player}{\Game}{\sMattr{\player}{\Game}{X}{c} \cup X} \cap \cpre{\otherplayer}{\Game}{A_i} = \emptyset$.
  Since $X \subseteq \sMattr{\player}{\Game}{X}{c}$ we have that $\sMattr{\player}{\Game}{X}{c} \cup X = \sMattr{\player}{\Game}{X}{c}$, and therefore it suffices to show that $\spre{\player}{\Game}{\sMattr{\player}{\Game}{X}{c}} \cap \cpre{\otherplayer}{\Game}{A_i} = \emptyset$.
  Consider any vertex $v \in \spre{\player}{\Game}{\sMattr{\player}{\Game}{X}{c}}$. We distinguish two cases. 
  If $v \in V_\player$ then we have that $vE \cap \sMattr{\player}{\Game}{X}{c} \neq \emptyset$ and thus $v \not\in (V_\player \setminus (\pre{\Game}{V \setminus A_i} \cup \sinks{\Game}))$.
  Otherwise, $v \in V_{\otherplayer}$ and we have that $vE \subseteq \sMattr{\player}{\Game}{X}{c}$ and therefore $v \notin V_{\otherplayer} \cap \pre{\Game}{A_i}$.
  
  Thus we conclude that $\sMattr{\player}{\Game}{X}{c} \cap \cpre{\otherplayer}{\Game}{A_i} = \emptyset$. 
  \qed
\end{proof}

We conclude with the correctness of safe fatal attractors.

\begin{proposition}\label{prop:fatal}
  Let $\Game$ be an incomplete parity game.
  We have $\sfattr{\peven}{\Game}{c} = \fattr{\peven}{\Game}{c}$ for even $c$ and for odd $c$ we have $\sfattr{\podd}{\Game}{c} = \fattr{\podd}{\Game}{c}$.
\end{proposition}

\begin{proof}
The proof proceeds along the lines of that of Proposition~\ref{prop:forced_cycles}.
Let $c$ be an even priority and let $\tau(Z) = P^{=c} \cap \sMattr{\peven}{\Game}{Z}{c}$.
We show that $\fattr{\peven}{\Game}{c}$ is indeed a fixed point of $\tau$.

By definition of $\fattr{\peven}{\Game}{c}$, we know that $\fattr{\peven}{\Game}{c}$ is equal to $P^{=c} \cap \safe{\peven}{\Game} \cap \Mattr{\peven}{\Game \cap \safe{\peven}{\Game}}{\fattr{\peven}{\Game}{c}}{c}$.
Since $\fattr{\peven}{\Game}{c} \subseteq \safe{\peven}{\Game}$ we know that the monotone attractor $\Mattr{\peven}{\Game \cap \safe{\peven}{\Game}}{\fattr{\peven}{\Game}{c}}{c} = \sMattr{\peven}{\Game}{\fattr{\peven}{\Game}{c}}{c}$ by Lemma~\ref{lemma:safe_monotone_attractor}.
Furthermore, $\sMattr{\peven}{\Game}{\fattr{\peven}{\Game}{c}}{c} \subseteq \safe{\peven}{\Game}$ by Lemma~\ref{lemma:safe_monotone_attractor_safe} and thus $\fattr{\peven}{\Game}{c} = P^{=c} \cap \sMattr{\peven}{\Game}{\fattr{\peven}{\Game}{c}}{c} = \tau(\fattr{\peven}{\Game}{c})$.

Next, we show that for any $Z = \tau(Z)$ it holds that $Z \subseteq \fattr{\peven}{\Game}{c}$. Let $Z$ be such.
By Lemma~\ref{lemma:safe_monotone_attractor_safe} we have that $Z \subseteq \safe{\peven}{\Game}$ since $Z \subseteq \sfattr{\peven}{\Game}{c} \subseteq \sMattr{\peven}{\Game}{\sfattr{\peven}{\Game}{c}}{c}$.
Therefore, by Lemma~\ref{lemma:safe_monotone_attractor} it follows that $\sMattr{\peven}{\Game}{Z}{c} = \Mattr{\peven}{\Game \cap \safe{\peven}{\Game}}{Z}{c}$.
Thus we conclude that $Z = P^{=c} \cap \safe{\peven}{\Game} \cap \Mattr{\peven}{\Game \cap \safe{\peven}{\Game}}{Z}{c}$ and therefore $Z \subseteq \fattr{\peven}{\Game}{c}$.

The proof for odd priority $c$ is completely dual.
\qed
\end{proof}

Similar to algorithm \texttt{solB$^-$}, the algorithm \texttt{solB$^-_s$} computes safe fatal attractors for priorities in descending order and collects the safe-$\player$-attractor extended dominions obtained this way.

\section{Experimental Results}\label{section:experiments}

We experimentally evaluate the techniques of Section~\ref{sec:partial_solvers}.
For this, we use games stemming from practical model checking and equivalence checking problems.
Our experiments are run, single-threaded, on an Intel~Xeon~6136~CPU~@~3~GHz PC.
The sources for these experiments can be obtained from the downloadable artefact~\cite{artifact}.

\subsection{Implementation}

We have implemented a symbolic exploration technique for parity games in the mCRL2 toolset~\cite{BunteGKLNVWWW19}.
Our tool exploits techniques such as \emph{read} and \emph{write} dependencies~\cite{KantP14,BlomPW10}, and uses sophisticated exploration strategies such as \emph{chaining} and \emph{saturation}~\cite{CiardoMS06}.
We use MDD-like data structures~\cite{Miller93} called \emph{List Decision Diagrams (LDDs)}, and the corresponding Sylvan implementation~\cite{DijkP17}, to represent parity games symbolically.
Sylvan also offers efficient implementations for set operations and relational operations, such as predecessors, facilitating the implementation of attractor computations, the described (partial) solvers, and a full solver based on Zielonka's recursive algorithm~\cite{Zielonka98}, which remains one of the most competitive algorithms in practice, both explicitly and symbolically~\cite{SanchezWW18,Dijk18}.
For the attractor set computation we have also implemented chaining to determine (multi-)step $\player$-predecessors more efficiently.

For all three on-the-fly solving techniques of Section~\ref{sec:partial_solvers}, we have implemented 1)~a variant that runs the standard (partial) solver on the $\player$-safe subgame and removes the found dominion using the standard attractor (within that subgame), and 2)~a variant that uses (partial) solvers with the safe attractors. 
Moreover, we also conduct experiments using the full solver running on an $\player$-safe subgame.
An important design aspect is to decide how the exploration and the on-the-fly solving should interleave.
For this we have implemented a time based heuristic that keeps track of the time spent on solving and exploration steps.
The time measurements are used to ensure that (approximately) ten percent of total time is spent on solving by delaying the next call to the solver.
We do not terminate the partial solver when it requires more time, and thus it is only approximate.
As a result of this heuristic, cheap solvers will be called more frequently than more expensive (and more powerful) ones, which may cause the latter to explore larger parts of the game graph.

\subsection{Cases}

Table~\ref{table:experiments} provides an overview of the models and a description of the property that is being checked.
The properties are written in the modal $\mu$-calculus with data~\cite{GrooteW05a}.
For the equivalence checking case we have mutated the original model to introduce a defect.
For each property, we indicate the \emph{nesting depth} (ND) and \emph{alternation depth}~\cite{CleavelandKS92} and whether the parity game is \emph{solitaire} (Yes/No).
The nesting depth indicates how many different priorities occur in the resulting game; for our encoding this is at most ND+2 (the additional ones encode constants `\emph{true}' and `\emph{false}').
The alternation depth is an indication of a game's complexity due to alternating priorities.

\begin{table}[h]
\centering
\scriptsize
\caption{Models and formulas.}\label{table:experiments}
\begin{tabular}{ l l c c c c c l }
  Model & Ref. & Prop. & Result & ND & AD & Sol. & Description \\
  \toprule
  SWP & \cite{TanenbaumW11} & 1 & \textsf{false} & 1 & 1 & Y & No error transition \\ 
      &  & 2 & \textsf{false} & 3 & 3 & N & Infinitely often enabled then infinitely often taken
\\ 
  WMS & \cite{RemenskaWVTB13} & 1 & \textsf{false} & 1 & 1 & Y & Job failed to be done \\ 
      &   & 2 & \textsf{false} & 1 & 1 & Y & No zombie jobs \\ 
      &   & 3 & \textsf{true} & 3 & 2 & Y &  A job can become alive again infinitely often \\ 
      &   & 4 & \textsf{false} & 2 & 2 & N & Branching bisimulation with a mutation \\ 
  
  BKE & \cite{BlomGMS05} & 1 & \textsf{true} &  1& 1 & Y & No secret leaked \\ 
  &  & 2 & \textsf{false} & 2 & 1 & N & No deadlock \\ 
  CCP & \cite{PangFHV07}  & 1 & \textsf{false} & 2 & 1 & N & No deadlock  \\
      &                   & 2 & \textsf{false} & 2 & 1 & N & After access there is always accessover possible \\ 
  PDI & n/a & 1 & \textsf{true} & 2 & 1 & N & Controller reaches state before it can connect again \\ 
      & & 2 & \textsf{false} & 2 & 1 & N & Connection impermissible can always happen or we \\ 
      & &   &                & & & & establish a connection \\ 
      & & 3 & \textsf{false} & 3 & 1 & N & When connected move to not ready for connection
and \\ 
      & &   &                & & & & do not establish a connection until it is allowed again \\ 
      & & 4 & \textsf{true}  & 2 & 1 & N & The interlocking moves to the state connection closed \\
      & &   &                & & & & before it is allowed to succesfully establish a connection \\  
  \bottomrule
\end{tabular}
\end{table}

We use MODEL-$i$ to indicate the parity game belonging to model MODEL and property $i$.
Models SWP, BKE and CCP are protocol specifications.
The model PDI is a specification of a EULYNX SCI-LX SySML interface model that is used for a train interlocking system.
Finally, WMS is the specification of a workload management system used at CERN.
Using tools in mCRL2~\cite{BunteGKLNVWWW19}, we have converted each model and property combination into a so-called parameterised Boolean equation systems~\cite{GrooteW05}, a higher-level logic that can be used to represent the underlying parity game.

Parity games SWP-$1$, WMS-$1$, WMS-$2$ and BKE-$1$ encode typical safety properties where some action should not be possible.
In terms of the alternation-free modal mu-calculus with regular expressions, such properties are of the shape $[\textsf{true}^*. \textsf{a}]\textsf{false}$.
These properties are violated exactly when the vertex encoding `\emph{false}' can be reached.
Parity games SWP-$2$, WMS-$3$ and WMS-$4$ are more complex properties with alternating priorities, where WMS-$4$ encodes branching bisimulation using the theory presented in~\cite{ChenPPW07}.
The parity games BKE-$2$ and CCP-$1$ encode a `no deadlock' property given by a formula which states that after every path there is at least one outgoing transition.
Finally, CCP-$2$ and all PDI cases contain formulas with multiple fixed points that yield games with multiple priorities but no (dependent) alternation.
\begin{table}[H]
\centering
\scriptsize
\caption{Experiments with parity games where on-the-fly solving cannot terminate early.
All run times are in seconds. 
The number of vertices is given in millions. 
Memory is given in gigabytes.
Bold-faced numbers indicate the lowest value.}\label{table:full_exploration}
\begin{tabular}{ l l r r r r r}
 Game & Strategy & Vertices ($10^6$) & Explore (s) & Solve (s) & Total (s) & Mem (GB) \\
  \toprule
BKE-1 & full & 40 & 640  & 65 & \textbf{705} & 14 \\
      & solitaire & 40/40 & 629/615  & 153/100 & 782/715 & 15/15 \\
      & cycles & 40/40 & 635/644 & 149/160 & 785/804 & 15/15 \\
      & fatal & 40/40 & 624/625 & 152/164 & 776/789 & 15/15 \\
      & partial & 40 & 651  & 147 & 798 & 15 \\
PDI-1 & full & 114 & 27  & 0.1 & \textbf{28} & 2 \\
      & solitaire & 114/114 & 28/27 & 4/0 & 33/\textbf{28} & 2/2 \\
      & cycles & 114/114 & 29/28 & 7/7 & 36/35 & 2/2 \\
      & fatal & 114/114 & 28/28  & 4/7 & 32/35 & 2/2 \\
      & partial & 114 & 28  & 9 & 37 & 2 \\
PDI-4 & full & 474 & 286  & 0 & \textbf{287} & 2 \\
      & solitaire & 474/474 & 284/281  & 46/14 & 331/295 & 2/2 \\
      & cycles & 474/474 & 284/287  & 92/91 & 376/378 & 2/2 \\
      & fatal & 474/474 & 285/283  & 80/91 & 365/374 & 2/2 \\
      & partial & 474 & 286  & 64 & 350 & 2 \\
  \bottomrule
\end{tabular}
\end{table}

\subsection{Results}

In Tables~\ref{table:full_exploration} and~\ref{table:early_termination} we compare the on-the-fly solving strategies presented in Section~\ref{sec:partial_solvers}.
In the `Strategy' column we indicate the on-the-fly solving strategy that is used.
Here \emph{full} refers to a complete exploration followed by solving with the Zielonka recursive algorithm.
We use \emph{solitaire} to refer to solitaire winning cycle detection, \emph{cycles} for forced winning cycle detection, \emph{fatal} to refer to fatal attractors and finally \emph{partial} for on-the-fly solving with a Zielonka solver on safe regions.
For solvers with a standard variant and a variant that utilises the safe attractors the first number indicates the result of applying the (standard) solver on \emph{safe} vertices, and the second number (following the slash `$/$') indicates the result when using the solver that utilises safe attractors.
\begin{table}[p]
\centering
\scriptsize
\caption{Experiments with parity games in which \emph{at least one} partial solver terminates early.
All run times are in seconds. 
The number of vertices is given in millions.
For solvers with two variants the first number indicates the result of applying the solver on \emph{safe} vertices, and following the slash `$/$' the result when using the solver that uses safe attractors.
Memory is given in gigabytes.
Bold-faced numbers indicate the lowest value.}\label{table:early_termination}
\begin{tabular}{ l l r r r r r}
 Game & Strategy & Vertices ($10^6$) & \hspace{4mm} Explore (s) & \hspace{4mm} Solve (s) & \hspace{5mm} Total (s) & Mem (GB) \\
  \toprule
SWP-1 & full & 13304 & \timeout & n/a & \timeout & \timeout \\
      & solitaire & 15.1/0.4 & 8.5/1.4  & 27.3/0.1 & 35.8/\textbf{1.5} & 2.8/1.5 \\
      & cycles & 25.2/0.9 & 12.3/1.8  & 42.7/1.0 & 55.0/2.8 & 3.2/1.5 \\
      & fatal & 15.1/0.4 & 9.0/1.3  & 29.4/0.4 & 38.4/1.7 & 3.1/1.5 \\
      & partial & 27.1 & 13.1 & 50.4 & 63.5 & 3.6 \\
SWP-2 & full & 1987 & \timeout  & n/a & \timeout & \timeout \\
      & solitaire & 1631/1987 & \timeout/\timeout  & 163/11 & \timeout/\timeout & \timeout/\timeout \\
      & cycles & 1774/1774 & \timeout/\timeout  & 154/91 & \timeout/\timeout & \timeout/\timeout \\
      & fatal & 0.007/0.007 & 0.9/0.9 & 0.4/0.2 & 1.3/\textbf{1.0} & 1.4/1.2 \\
      & partial & 0.007 & 0.9 & 0.4 & 1.3 & 1.4 \\
WMS-1 & full & 270 & 2.8  & 0.4 & 3.3 & 0.2 \\
      & solitaire & 270/240 & 2.8/2.5 & 0.8/0.4 & 3.6/\textbf{2.9} & 0.3/0.2 \\
      & cycles & 270/270 & 2.9/3.2 & 0.8/8.0 & 3.7/11.2 & 0.3/0.5 \\
      & fatal & 270/270 & 2.6/3.2 & 0.8/8.5 & 3.4/11.7 & 0.3/0.5 \\
      & partial & 270 & 2.7 & 0.8 & 3.5 & 0.3 \\
WMS-2 & full & 317 & 3.3  & 0.3 & 3.6 & 0.2 \\
      & solitaire & 7/7 & 0.2/0.2 & 1.0/0.5 & 1.2/\textbf{0.8} & 0.1/0.1 \\
      & cycles & 7/66 & 0.2/0.8 & 1.0/2.7 & 1.2/3.4 & 0.1/0.2 \\
      & fatal & 7/66 & 0.2/0.7 & 1.0/2.9 & 1.3/3.6 & 0.1/0.2 \\
      & partial & 7 & 0.2 & 1.1 & 1.3 & 0.1 \\
WMS-3 & full & 317 & 2.6 & 0.1 & 2.7 & 0.2 \\
      & solitaire & 317/317 & 2.6/2.6 & 0.4/0.3 & 3.1/2.9 & 0.2/0.2 \\
      & cycles & 317/317 & 2.7/2.7  & 0.4/0.6 & 3.1/3.3 & 0.2/0.2 \\
      & fatal & 5/1 & 0.2/0.1  & 0.5/0.1 & 0.7/\textbf{0.2} & 0.1/0.1 \\
      & partial & 5 & 0.2  & 0.3 & 0.5 & 0.1 \\
WMS-4 & full & 366 & \timeout & n/a & \timeout & \timeout \\
      & solitaire & 0.03/0.03 & 38/38  & 0.8/0.1 & 39/38 & 2/2 \\
      & cycles & 0.03/0.03 & 37/37  & 0.8/0.3 & 38/\textbf{37} & 2/2 \\
      & fatal & 0.03/0.03 & 37/37  & 0.8/0.3 & 38/\textbf{37} & 2/2 \\
      & partial & 0.03 & 37  & 0.7 & 38 & 2 \\
BKE-2 & full & 119 & 942 & 36.5 & 979 & 28 \\
      & solitaire & 0.0007/0.0001 & 0.2/0.1 & 0.0/0.0 & \textbf{0.2}/\textbf{0.2} & 0.9/0.9 \\
      & cycles & 0.0007/0.0003 & 0.2/0.2 & 0.0/0.0 & \textbf{0.2}/\textbf{0.2} & 0.9/0.9 \\
      & fatal & 0.0007/0.0003 & 0.2/0.2  & 0.0/0.0 & \textbf{0.2}/\textbf{0.2} & 0.9/0.9 \\
      & partial & 0.0007 & 0.2 & 0.0 & \textbf{0.2} & 0.9 \\
CCP-1 & full & 0.4 & 28  & 4.2 & 32 & 2 \\
      & solitaire & 0.003/0.003 & 1.0/1.0 & 0.1/0.1 & \textbf{1.1}/\textbf{1.1} & 2/2 \\
      & cycles & 0.003/0.003 & 1.0/1.0  & 0.1/0.1 & \textbf{1.1}/\textbf{1.1} & 2/2 \\
      & fatal & 0.006/0.003 & 1.3/1.1  & 0.1/0.1 & 1.4/1.2 & 1.5/1.5 \\
      & partial & 0.003 & 1.0  & 0.1 & \textbf{1.1} & 1.5 \\
CCP-2 & full & 0.9 & 35 & 33 & 68 & 1.7 \\
      & solitaire & 0.02/0.007 & 1.6/1.1 & 0.2/0.0 & 1.8/\textbf{1.1} & 1.5/1.5 \\
      & cycles & 0.02/0.007 & 1.9/1.1  & 0.2/0.1 & 2.1/1.2 & 1.5/1.5 \\
      & fatal & 0.02/0.007 & 1.6/1.2  & 0.2/0.1 & 1.8/1.3 & 1.5/1.5 \\
      & partial & 0.02 & 1.6  & 0.2 & 1.8 & 1.5 \\
PDI-2 & full & 229 & 31  & 12 & 43 & 2 \\
      & solitaire & 229/229 & 33/32 & 34/12 & 67/45 & 2/2 \\
      & cycles & 30/30 & 15/14 & 3/5 & \textbf{17}/19 & 2/2 \\
      & fatal & 30/30 & 15/15  & 3/5 & 18/19 & 2/2 \\
      & partial & 123 & 23  & 29 & 51 & 2 \\
PDI-3 & full & 436 & 228  & 8 & 236 & 2 \\
      & solitaire & 436/436 & 230/228 & 36/32 & 266/260 & 2/2 \\
      & cycles & 78/162 & 65/102 & 19/64 & 84/166 & 2/2 \\
      & fatal & 75/84 & 64/67  & 19/23 & \textbf{83}/90 & 2/2 \\
      & partial & 110 & 82  & 30 & 112 & 2 \\
  \bottomrule
\end{tabular}
\end{table}

The column `Vertices' indicates the number of vertices explored in the game.
In the next columns we indicate the time spent on exploring and solving specifically and the total time in seconds. 
We exclude the initialisation time that is common to all experiments.
Finally, the last column indicates memory used by the tool in gigabytes.
We report the average of 5~runs and have set a timeout (indicated by $\timeout$) at 1200 seconds per run.
Table~\ref{table:full_exploration} contains all benchmarks that require a full exploration of the game graph, providing an indication of the overhead in cases where this is unavoidable; Table~\ref{table:early_termination} contains all benchmarks where \emph{at least one} of the partial solvers allows exploration to terminate early.

For games SWP-$1$, WMS-$1$, WMS-$2$ in Table~\ref{table:early_termination} we find that \emph{solitaire}, and in particular the safe attractor variant, is able to determine the solution the fastest.
Also, for all entries in Table~\ref{table:full_exploration} this is the solver with the least overhead.
Next, we observe that for cases such as WMS-$1$ and PDI-$3$ using the safe attractor variants of the solvers can be detrimental.
Our observation is that first computing safe sets (especially using chaining) can be quick when most vertices are owned by one player and one priority and the computation of the safe attractor, which uses the more difficult safe control predecessor is more involved in such cases.
There are also cases WMS-$3$, WMS-$4$, CCP-$1$ and CCP-$2$ where the safe attractor variants are faster and these cases all have multiple priorities.
In cases where these solvers are slow (for example PDI-$3$) we also observe that more states are explored before termination, because the earlier mentioned time based heuristic results in calling the solver significantly less frequently.

For parity games SWP-2 and WMS-3 only \emph{fatal} and \emph{partial} are able to find a solution early, which shows that more powerful partial solvers can be useful.
From Table~\ref{table:full_exploration} and the cases in which the safe attractor variants perform poorly we learn that the partial solvers can, as expected, cause overhead.
This overhead is in our benchmarks on average 30 percent, but when it terminates early it can be very beneficial, achieving speed-ups of up to several orders of magnitude.

\section{Conclusion}

In this work we have developed the theory to reason about on-the-fly solving of parity games, independent of the strategy that is used to explore games.
We have introduced the notion of \emph{safe} vertices, shown their correctness, proven an optimality result, and we have studied partial solvers and shown that these can be made to run without determining the safe vertices first; which can be useful for on-the-fly solving.
Finally, we have demonstrated the practical purpose of our method and observed that solitaire winning cycle detection with safe attractors is almost always beneficial with minimal overhead, but also that more powerful partial solvers can be useful.

Based on our experiments, one can make an educated guess which partial solver to select in particular cases; we believe that this selection could even be steered by analysing the parameterised Boolean equation system representing the parity game.
It would furthermore be interesting to study (practical) improvements for the safe attractors, and their use in Zielonka's recursive algorithm.

\subsubsection{Acknowledgements}

We would like to thank Jeroen Meijer and Tom van Dijk for their help regarding the Sylvan library when implementing our prototype.
This work was supported by the TOP Grants research programme with project number 612.001.751 (AVVA), which is (partly) financed by the Dutch Research Council (NWO).

\newpage
\bibliographystyle{splncs04}
\bibliography{bibliography}

\newpage
\appendix

\end{document}